\documentclass[letterpaper, 10 pt, conference]{ieeeconf}  

\IEEEoverridecommandlockouts                              
\usepackage{graphics} 
\usepackage{epsfig} 
\usepackage{mathptmx} 
\usepackage{times} 
\usepackage{amsmath} 
\usepackage{amssymb}  
\usepackage{cite}   
\usepackage{xcolor}
\usepackage{tikz}
\usetikzlibrary{arrows.meta, positioning, calc}

\newtheorem{theorem}{Theorem}
\newtheorem{remark}{Remark}

\title{\LARGE \bf
Robust multi-scale leader–follower control of large multi-agent systems}
\author{Davide Salzano, Gian Carlo Maffettone and Mario di Bernardo
\thanks{Davide Salzano and Gian Carlo Maffettone contributed equally to this work. Corresponding author: Mario di Bernardo.}
\thanks{Davide Salzano and Mario di Bernardo are with the Department fo Electrical Engineering and Information Technologies of the University of Naples Federico II, Naples, Italy. ({\tt\small davide.salzano@unina.it, mario.dibernardo@unina.it})}%
\thanks{Gian Carlo Maffettone and Mario di Bernardo are with the Modeling and Engineering Risk and Complexity program of the Scuola Superiore Meridionale, Naples, Italy ({\tt\small giancarlo.maffettone@unina.it})}%
}

\begin{document}

\maketitle
\thispagestyle{empty}
\pagestyle{empty}

\begin{abstract}
In many multi-agent systems of practical interest, such as traffic networks or crowd evacuation, control actions cannot be exerted on all agents. Instead, controllable leaders must indirectly steer uncontrolled followers through local interactions. Existing results address either leader–follower density control of simple, unperturbed multi-agent systems or robust density control of a single directly actuated population, but not their combination. We bridge this gap by deriving a coupled continuum description for leaders and followers subject to unknown bounded perturbations, and designing a macroscopic feedback law that guarantees global asymptotic convergence of the followers' density to a desired distribution. The coupled stability of the leader–follower system is analyzed via singular perturbation theory, and an explicit lower bound on the leader-to-follower mass ratio required for feasibility is derived. Numerical simulations on heterogeneous biased random walkers validate our theoretical findings.
\end{abstract}

\section{Introduction}
Controlling the collective behavior of large groups of interacting agents is a central challenge in domains ranging from traffic management \cite{siri2021freeway} to crowd evacuation \cite{zhou2019guided} and swarm robotics \cite{dorigo2021swarm}. When the number of agents is large, designing individual controllers becomes intractable \cite{d2023controlling}. Macroscopic approaches overcome this limitation by reformulating the control problem at the density level, describing collective behavior through partial differential equations (PDEs) whose complexity is independent of the number of agents \cite{sinigaglia2022density, elamvazhuthi2023density}.

A particularly effective multi-scale approach \cite{diBernardo2026multiscale} bridging microscopic (agent-level) and macroscopic (density-level) descriptions of multi-agent 
systems  is {\em continuification} (or continuation) 
\cite{nikitin2021continuation,maffettone2022continuification}. 
In this approach, a continuum description is derived from 
the agent-based model, the controller is designed at the PDE 
level, and the resulting control field is mapped back to 
agent-level inputs. The link between the two levels of 
description relies on two operations: (i) a \emph{micro-to-macro 
bridge}, which estimates the macroscopic density from the 
agents' positions (e.g., via kernel density estimation), and (ii)
a \emph{macro-to-micro bridge}, which maps the control field 
computed at the PDE level back to individual agent inputs 
(e.g., via spatial sampling of the control field). This 
pipeline has been successfully applied to swarm robotics 
\cite{maffettone2024mixed}, traffic \cite{fueyo2025continuation}, 
and networks of lasers \cite{nikitin2021boundary}.

In many scenarios of practical interest, however, control actions cannot be exerted on all agents. In traffic networks, only autonomous vehicles are controllable \cite{siri2021freeway}; in crowd evacuation, only robotic guides can be actuated \cite{zhou2019guided}. This motivates the leader--follower paradigm, where a population of controllable leaders steers uncontrolled followers through local interactions. A continuification-based solution to this problem was proposed in \cite{maffettone2025leader} and extended to shepherding \cite{di2025continuification} and bio-inspired plasticity \cite{maffettone2025density}. A separate line of work \cite{maffettone2026robust} proposed a robust continuification strategy to counteract unknown bounded perturbations, but only for a single, directly actuated population.

The combination of these two directions, robustness and indirect actuation, poses challenges that neither framework addresses when operating in isolation. This gap is tackled in this work. When both populations are subject to unknown bounded perturbations, the followers' convergence depends on the ability of the leaders to track a time-varying reference density, generated by the robust controller. This couples the two populations' dynamics in a way that requires a dedicated stability analysis. In this paper, we address this problem by deriving a coupled continuum description for leaders and followers subject to bounded perturbations, and designing a macroscopic feedback control law that guarantees global asymptotic convergence of the followers' density to a desired distribution. We derive an explicit lower bound on the leader mass required for the control problem to admit a solution, and, using singular perturbation theory, we analyze the stability of the coupled system. Numerical simulations on heterogeneous biased random walkers validate our findings and illustrate robustness to increasing levels of population heterogeneity and finite populations size.

\begin{figure*}[t]
\centering
\resizebox{1.8\columnwidth}{!}{%
\begin{tikzpicture}[
    >=latex,
    font=\normalsize,
    node distance=0.9cm,
    block/.style={
        draw,
        rectangle,
        minimum height=1cm,
        align=center
    },
    sum/.style={
        draw,
        circle,
        minimum size=5mm,
        inner sep=0pt
    }
]
\node[sum] (sumF) {};
\node[block, minimum width=2.6cm, right=1.2cm of sumF] (refgen)
    {Followers control};
\node[sum, right=of refgen] (sumL) {};
\node[block, minimum width=2.1cm, right=of sumL] (controller)
    {Leaders'\\controller};
\node[block, minimum width=2.2cm, right=1.3cm of controller] (m2a)
    {Macro-to-micro\\bridge};
\node[block, minimum width=2.6cm, right=of m2a] (agents)
    {Multi-agent\\dynamics};
\node[block, minimum width=2.2cm, right=of agents] (a2m)
    {Micro-to-macro\\bridge};

\draw[->] ([xshift=-1.9cm]sumF.west) -- (sumF.west)
    node[midway, above] {$\bar{\rho}^{F}$};
\node at ([xshift=-4mm,yshift=-2mm]sumF.south) {$-$};
\draw[->] (sumF.east) -- (refgen.west)
    node[midway, above] {$e^{F}$};

\draw[->] (refgen.east) -- (sumL.west)
    node[midway, above] {$\bar{\rho}^{L}$};
\node at ([xshift=-3mm,yshift=-2mm]sumL.south) {$-$};
\draw[->] (sumL.east) -- (controller.west)
    node[midway, above] {$e^{L}$};

\draw[->] (controller.east) -- (m2a.west)
    node[midway, above] {$u(x,t)$};
\draw[->] (m2a.east) -- (agents.west)
    node[midway, above] {$u_i$};
\draw[->] (agents.east) -- (a2m.west)
    node[midway, above] {$\mathbf{x}_i$};

\coordinate (tapF) at ([xshift=1.4cm]a2m.east);
\coordinate (outF) at ([xshift=1.9cm]a2m.east);
\coordinate (outL) at ([xshift=0.9cm,yshift=-0.3cm]a2m.east);

\draw[-] (a2m.east) -- (outF)
    node[midway, above] {$\rho^{F}$};
\draw[->] (outF) -- ++(0.5cm,0);

\draw[-] ([yshift=-0.3cm]a2m.east) -- (outL)
    node[midway, below] {$\rho^{L}$};

\def\fbL{-1.4cm}
\def\fbF{-2.2cm}

\draw (outL) |- ([yshift=\fbL]sumL.south);
\draw[->] ([yshift=\fbL]sumL.south) -- (sumL.south);

\draw (tapF) |- ([yshift=\fbF]sumF.south);
\draw[->] ([yshift=\fbF]sumF.south) -- (sumF.south)
    node[midway, right] {$\rho^{F}$};

\end{tikzpicture}%
}
\caption{Block diagram of the multi-scale leader-follower control architecture. An outer loop regulates the followers' density by comparing the reference $\bar\rho^F$ with the estimated density $\rho^F$, generating a reference density $\bar\rho^L$ for the leaders. An inner loop tracks this reference, producing the macroscopic control field $u(x,t)$. The \emph{macro-to-micro bridge} maps $u(x,t)$ to individual leader inputs $u_i$ via spatial sampling. The multi-agent dynamics produce leader and follower positions $\mathbf{x}_i$, which the \emph{micro-to-macro bridge} converts back to estimated densities $\rho^L$ and $\rho^F$ via density estimation, closing both feedback loops.}
\label{fig:block_diagram}
\end{figure*}
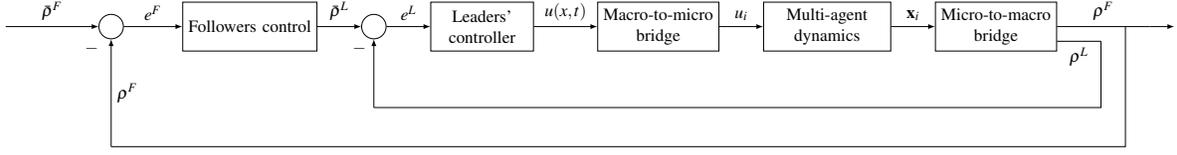

\section{Mathematical modeling}\label{sec:mathematical_modeling}
\subsection{Microscopic model}
We consider two populations of interacting agents, leaders and followers, evolving on the unit circle $\Omega= [-\pi, \pi]$. The dynamics of $N^L$ leaders and $N^F$ followers are
\begin{subequations}\label{eq:micro}
    \begin{align}
    \mathrm{d} x_i^L(t) &= \left[u_i(t)+h_i(t, X^L(t))\right]\,\mathrm{d}t, \quad i=1, \dots, N^L,\label{eq:leaders_micro}\\
    \mathrm{d}x_i^F(t) &= \left[\sum_{j=1}^{N^L} f(\{x_i^F(t),x_j^L(t)\}_\pi)+g_i(t, X^F(t))\right]\,\mathrm{d}t \nonumber\\&\quad\quad+ \sqrt{2D}\,\mathrm{d}W_i(t), \quad i = 1 ,\dots, N^F,\label{eq:followers_Equations}
\end{align}
\end{subequations}
where $x_i^L, x_i^F\in \Omega$ are the states of the $i$-th leader and follower, $X^L\in\Omega^{N_L}$ and $X^F\in\Omega^{N_F}$ are the stack vectors containing the states of all leaders and followers, $h_i: \mathbb{R}\times\Omega^{N^L}\to \mathbb{R}$ and $g_i: \mathbb{R}\times\Omega^{N^F}\to \mathbb{R}$ are internal dynamics characterizing each population, $u_i \in \mathbb{R}$ is the control input acting on leaders, and $W_i$ is a standard Wiener process with diffusion coefficient $D$. The interaction kernel $f:\Omega\to \mathbb{R}$ is the odd kernel \cite{maffettone2025leader}
\begin{equation} \label{eq:periodic_kernel}
    f(x) = \frac{\mathrm{sgn}(x)}{\mathrm{e}^{2\pi/\ell}-1} \left[\mathrm{e}^{\frac{2\pi-\vert x\vert }{\ell}} - \mathrm{e}^{\frac{\vert x\vert}{\ell}}\right].
\end{equation}
It describes how leaders influence followers, with $\ell$ being the characteristic interaction length. The internal dynamics $h_i$ and $g_i$ are unknown and uniformly bounded, that is, 
\begin{subequations}
    \begin{align}
    &\vert h_i(t, X^L(t))\vert < r, \, \forall i\in{1,\ldots,N^L},\, \forall t\in\mathbb{R}_{\geq 0},\\
    &\vert g_i(t, X^F(t))\vert < k,\, \forall i\in{1,\ldots,N^F},\, \forall t\in\mathbb{R}_{\geq 0},
\end{align}
\end{subequations}
with $r, k\geq 0$.
 
\subsection{Macroscopic model}
Since the perturbations are unknown but bounded, we apply the comparison principle to construct bounding systems that replace $h_i$ and $g_i$ with their worst-case values $\pm r$ and $\pm k$ \cite{maffettone2026robust}. Taking the mean-field limit of the bounding system of both populations yields the coupled PDEs
\begin{subequations}\label{eq:model}
\begin{align}
    \rho^L_t(x, t) &+ \left[\rho^L(x, t)  (u(x, t) \pm r )\right]_x = 0, \label{eq:leaders}\\
    \rho^F_t(x, t) &+ \left[\rho^F(x, t) (v^{FL}(x, t) \pm k)\right]_x = D\rho^F_{xx}(x, t),\label{eq:followers}
\end{align}
\end{subequations}
where $x\in\Omega$, the $\pm$ signs correspond to the upper and lower bounding systems, and
\begin{align}\label{eq:vFL}
    v^{FL}(x, t) = \int_{\Omega} f\left(\{x, y\}_\pi\right)\rho^L(y, t)\,\mathrm{d}y = (f*\rho^L)(x, t).
\end{align}
A control law that stabilizes both bounding systems guarantees stabilization of the original system~\eqref{eq:micro} for any realization of the perturbations. Model \eqref{eq:model} is complemented with initial conditions and periodic boundary conditions ensuring mass conservation
\begin{align}
    \left(\int_\Omega \rho^i(x, t)\,\mathrm{d}x\right)_t= M^i_t(t)=0, \quad i= L, F,
\end{align}
where $M^L,M^F$ are the mass of leaders and followers.

\subsection{Problem statement}\label{sec:prob_statement}
Given a time-invariant desired density profile $\bar{\rho}^F:\Omega \to \mathbb{R}_{> 0}$, such that $\int_\Omega \bar{\rho}^F\,\mathrm{d}x=M^F$, we seek for a periodic $u$ in \eqref{eq:model} such that
\begin{align}\label{eq:PS}
    \lim_{t\to\infty} \Vert e^F(\cdot, t) \Vert_2 =0,
\end{align}
where $e^F = \bar{\rho}^F-\rho^F$ and $\Vert \cdot \Vert_2$ is the $\mathcal{L}^2(\Omega)$-norm.

\section{Control design}\label{sec:control_design}
We design a control field $u:\Omega\times\mathbb{R}\to \mathbb{R}$ that guarantees asymptotic convergence of the control error. The strategy proceeds in two stages. First, we treat the velocity field $v^{FL}$ as a free design variable and derive stability conditions for the followers' dynamics. Then, we constrain $v^{FL} = f * \rho^L$ to recover a physically meaningful reference density for the leaders via deconvolution, and we construct a feedback law steering the leaders towards such a desired density. The resulting coupled system is analyzed using singular perturbation theory \cite{kokotovic1999singular}. Our strategy is complemented with micro-to-macro and macro-to-micro bridges as illustrated in Fig.~\ref{fig:block_diagram}: the micro-to-macro bridge consists in a density estimation procedure, while the macro-to-micro bridge is implemented via a spatial sampling.

\subsection{Followers control design} \label{sec:Followers_behavior}

As mentioned above, to study the stability properties of the follower population, we assume it is possible to arbitrarily choose their drift. 

Under this assumption, we choose $v^{FL}$ as the convex combination
\begin{align}\label{eq:control_action}
    v^{FL}(x, t) = [1-\alpha(t)] v^{FF}(x) + \alpha(t)v^{FB}(x, t),
\end{align}
where $\alpha:\mathbb{R}_{\geq0}\to[0, 1]$ is a design function, $v^{FF}$ is a feedforward action and $v^{FB}$ is a feedback correction.
More precisely, $v^{FF}$ is chosen as in \cite{maffettone2025leader}, that is
\begin{equation}
    v ^{FF}(x) = D\frac{\bar \rho^F_x(x)}{\bar \rho^F(x)}.
\end{equation}
This choice ensures that, in absence of any disturbance and unmodeled dynamics, when $v^{FB}=0$, the followers displace according to the desired density profile if $\|g_1\|_\infty<2$, with 
\begin{align}
        g_1(x) &=  \left[\frac{\bar{\rho}_x^F(x)}{\bar{\rho}^F(x)}\right]_x.\label{eq:G1}
\end{align}
For more details on this result, see \cite{maffettone2025leader}.

\begin{theorem} \label{thm:followers_convergence}
    We choose $v^{FL}$ in \eqref{eq:control_action}, with $v^{FB}$ such that
    \begin{align}\label{eq:vfl_tilde}
         \left[\rho^F(x, t)v^{FB}(x, t)\right]_x = q^F(x, t).
    \end{align}
    Here
    \begin{align}\label{eq:q}
        q^F(x, t) = -k_p^Fe^F(x, t) - k_s^F(t)\mathrm{sign}\left[e^F(x, t)\right] + \beta(t),
    \end{align}
    where $\beta$ is a bounded function of time. If $\Vert g_1 \Vert_\infty<2$, $k_p^F > 0$,
    \begin{align} \label{eq:choice_ks}
        k_s^F(t) > \frac{\alpha(t) D\Vert \bar{\rho}^F_{xx}(\cdot)\Vert_\infty + k\Vert \bar{\rho}^F_x(\cdot)\Vert_\infty}{\alpha(t)}, 
    \end{align}
    then $e^F$ globally asymptotically converges to 0 in $\mathcal{L}^2(\Omega)$, $\forall \alpha\in(0,1]$.
\end{theorem}
\begin{proof}
    We substitute \eqref{eq:control_action} with \eqref{eq:vfl_tilde} into \eqref{eq:followers}, yielding (dropping dependence from space and time for compactness)
    \begin{align}\label{eq:rho_f_t}
        \rho^F_t + (1-\alpha) D\left(\rho^F\frac{\bar{\rho}^F_{xx}}{\bar{\rho}^F}\right)_x + \alpha q^F\pm k\rho^F_x = D\rho^F_{xx}. 
    \end{align}
    Recalling that $e^F = \bar{\rho}^F-\rho^F$, we can rephrase \eqref{eq:rho_f_t} in terms of the error function, resulting in
    \begin{multline}\label{eq:eFt}
        e^F_t = -\alpha D\bar{\rho}^F_{xx} - (1-\alpha)D  \left(e^F\frac{\nabla\bar{\rho}^F}{\bar{\rho}^F}\right)_x + \alpha q^F \\\pm k\bar{\rho}^F_x \mp k e^F_x + De^F_{xx}.
    \end{multline}
    Initial and periodic boundary conditions of the error system can be recovered from those of \eqref{eq:model}. We introduce the Lyapunov functional $V^F=\frac{1}{2}\Vert e^F\Vert_2^2$ and compute its time derivative
    \begin{multline}\label{eq:Vt}
        V^F_t = \int_{\Omega}e^Fe^F_t\,\mathrm{d}x = -D\alpha \int_\Omega e^F \bar{\rho}^F_{xx} \,\mathrm{d}x  +D  \int_\Omega e^F e^F_{xx} \,\mathrm{d}x \\-(1-\alpha)D \int_\Omega e^F \left(e^F\frac{\bar{\rho}^F_x}{\bar{\rho}^F}\right)_x\,\mathrm{d}x \pm k\int_\Omega e^F\bar{\rho}^F_x\,\mathrm{d}x \\\mp k \int_\Omega e^Fe^F_x\,\mathrm{d}x + \alpha \int_\Omega e^F q^F\,\mathrm{d}x,
    \end{multline}
    where we used \eqref{eq:eFt}.

    Let us establish the following relations:
    \begin{subequations}\label{eq:bounds}
        \begin{align}
            &D\int_\Omega e^Fe^F_{xx} \,\mathrm{d}x = -D \int_\Omega \left(e^F_{x}\right)^2 \,\mathrm{d}x = - D \Vert e^F_x \Vert_2^2 \leq -2DV^F,\\
            &\alpha D\left\vert \int_{\Omega}e^F\bar{\rho}^F_{xx}\,\mathrm{d}x\right\vert \leq \alpha D\Vert e^F\bar{\rho}^F_{xx} \Vert_1 \leq \alpha D \Vert e^F \Vert_1 \Vert \bar{\rho}^F_{xx}\Vert_\infty,\\
            &-(1-\alpha)D \int_\Omega e^F \left(e^F\frac{\bar{\rho}^F_x}{\bar{\rho}^F}\right)_x\,\mathrm{d}x\nonumber\\&= \frac{(1-\alpha) D}{2}\int_\Omega \left[(e^F)^2\right]_x \frac{\bar{\rho}^F_x}{\bar{\rho}^F}\,\mathrm{d}x\nonumber\\
            &= -\frac{(1-\alpha) D}{2}\int_\Omega (e^F)^2 g_1\,\mathrm{d}x \leq \frac{(1-\alpha) D}{2}\left\vert \int_\Omega (e^F)^2 g_1\,\mathrm{d}x\right\vert\nonumber\\
            &\leq \frac{(1-\alpha) D}{2}\Vert e^Fe^Fg_1\Vert_1 \leq (1-\alpha) D \Vert g_1\Vert_\infty V^F,\\
            &\pm k\int_{\Omega} e^F \bar{\rho}^F_x\,\mathrm{d}x\leq k\left\vert\int_{\Omega} e^F \bar{\rho}^F_x\,\mathrm{d}x \right\vert \leq k \Vert e^F \bar{\rho}^F_x\Vert_1 \nonumber\\&\leq k \Vert \bar{\rho}^F_x\Vert_\infty\Vert e^F \Vert_1\\
            &\pm k \int_{\Omega} e^F e^F_x \,\mathrm{d}x = \pm\frac{k}{2}\int_{\Omega} \left[(e^F)^2\right]_x  \,\mathrm{d}x = \pm \frac{k}{2} \left[(e^F)^2\right]_{-\pi}^\pi= 0,
            \end{align}
    \end{subequations}
    in which we used integration by parts (boundary terms vanish due to boundary conditions), Poincarè-Wirtinger \cite{heinonen2001lectures} and H$\mathrm{\Ddot{o}}$lder inequality \cite{axler2020measure}, and the identity $[(e^F)^2]_x = 2e^Fe^F_x$. Note that the last relation is true due to periodic boundaries, and the function $g_1$ is defined in \eqref{eq:G1}.
    
    Applying the bounds in \eqref{eq:bounds} to \eqref{eq:Vt} yields
    \begin{multline}\label{eq:penultimo_bound}
        V^F_t\leq \left[D(1-\alpha)\Vert g_1\Vert_\infty-2D\right]V^F\\ + \left[\alpha D \Vert\bar{\rho}^F_{xx} \Vert_\infty + k\Vert \bar{\rho}^F_x\Vert_\infty\right]\int_{\Omega}\vert e^F\vert\,\mathrm{d}x + \alpha\int_\Omega e^F q^F\,\mathrm{d}x.
    \end{multline}
    Substituting the expression of $q^F$ in \eqref{eq:q} into \eqref{eq:penultimo_bound} yields
    \begin{multline}\label{eq:Vt_bound}
        V^F_t\leq \left[D(1-\alpha)\Vert g_1\Vert_\infty-2D-\alpha k_p^F\right]V^F\\ + \left[\alpha D \Vert\bar{\rho}^F_{xx} \Vert_\infty + k\Vert \bar{\rho}^F_x\Vert_\infty-\alpha k_s^F\right]\int_{\Omega}\vert e^F\vert\,\mathrm{d}x.
    \end{multline}
    Under the theorem hypotheses, the right-hand side is negative proving the claim.
\end{proof}
\begin{remark}
    The function $\beta(t)$ in \eqref{eq:q} does not play any role for stability properties. However, it needs to be chosen to ensure the control flux fulfills boundary conditions; see \cite[Sec. VI.4]{maffettone2026robust}
\end{remark}
\begin{remark}
    The presence of discontinuous terms in \eqref{eq:q} and \eqref{eq:ql} (i.e., $\mathrm{sign}(\cdot)$) implies that the solutions of the closed-loop PDEs \eqref{eq:model} must be interpreted in a weak sense. Existence of such solutions under discontinuous flux is guaranteed by \cite[Definition~1]{cristofaro2019robust}; see also \cite[Remark~2]{maffettone2026robust} for a detailed discussion in  our setting.
\end{remark}

\subsection{Deconvolution and choice of $\alpha(t)$}\label{sec:deconv}
Since $v^{FL} = f * \rho^L$, we now recover a reference density $\bar{\rho}^L$ for the leaders that produces $v^{FL}$ as defined in \eqref{eq:control_action}. Using the kernel \eqref{eq:periodic_kernel}, any velocity field $v$ can be deconvolved to retrieve the generating density as
\begin{align} \label{eq:deconvolution}
    \rho(x, t) = \frac{v_x(x,t)}{2} - \frac{1}{2\ell^2}\int v(x,t)\,\mathrm{d}x+ I(t),
\end{align}
where $I:\mathbb{R}\to \mathbb{R}$ is an arbitrary function of time; see \cite[App.~B]{maffettone2025leader}. Since $v^{FL}$ is a convex combination of feedforward and feedback contributions and convolution is linear, the generating density takes the form
\begin{align}\label{eq:rho_hat_L}
    \bar{\rho}^L(x, t) = [1-\alpha(t)] \rho^{FF}(x) + \alpha(t) \rho^{FB}(x, t) +  G(t),
\end{align}
where $G:\mathbb{R}_{\geq 0}\to \mathbb{R}$ is an arbitrary function of time, and $\rho^{FF}$, $\rho^{FB}$ are obtained by deconvolving $v^{FF}$ and $v^{FB}$, respectively.

As noted in \cite{maffettone2025leader}, the deconvolution operation does not guarantee that the obtained densities are physically meaningful (i.e. that they are non negative and that they sum to the mass of available leaders). 
To guarantee positiveness, we choose $G(t) = - (1-\alpha(t)) \min_x(\rho^{FF}(x,t)) -  \alpha(t)\min_x(\rho^{FB}(x,t)) + C(t)$ that guarantees $\bar \rho^L(x,t)>0$ for any $C(t) \ge 0$.
Instead, to guarantee that the density integrates to the available mass of leaders $M^L$, we choose suitable values of $\alpha$ and $C$.
Specifically, it holds that
\begin{align}
    \int_\Omega\bar{\rho}^L\,\mathrm{d}\mathbf{x} = [1-\alpha] M^{FF} + \alpha M^{FB} + \vert \Omega\vert C,
\end{align}
where $M^{FB} = \int_\Omega \rho^{FB}\,\mathrm{d}x$ and we recall $M^{FF} = \int_\Omega \rho^{FF}\,\mathrm{d}x$. We need to choose $\alpha$ that guarantees 
\begin{align} \label{eq:alpha_feasibility}
    [1-\alpha] M^{FF} + \alpha M^{FB}(t) + \vert \Omega\vert C \leq M^L.
\end{align}
This allows choosing a non-negative $C$ such that $\int_\Omega\bar{\rho}^L\,\mathrm{d}x=M^L$. 
Note that \eqref{eq:alpha_feasibility} is equivalent to 
\begin{align}
    \alpha [M^{FB}-M^{FF}] + \vert \Omega\vert C \leq M^L - M^{FF}.
\end{align}
Assuming $M^L - M^{FF}>0$, which is guaranteed under the feasibility conditions reported in \cite{maffettone2025leader}, we choose
\begin{equation} \label{eq:alpha_choice}
    \alpha =
    \begin{cases}
        \frac{M^L - M^{FF}}{M^{FB}(t)-M^{FF}}\, &\mathrm{if}\, M^{FB}(t)>M^{FF}\\
        1, &\mathrm{otherwise}
    \end{cases}
\end{equation}
and 
\begin{equation}
    C = \max\left[0, M^L- M^{FF}-\alpha (M^{FB}-M^{FF})\right].
\end{equation}
This ensures that $\bar \rho ^L$ is a physically meaningful density.

\subsection{Minimum leader mass for feasibility}
The choice of $\alpha$ in \eqref{eq:alpha_choice} maximizes the feedback weight given the available leader mass, but can violate \eqref{eq:choice_ks}. To guarantee simultaneously a feasible $\bar{\rho}^L$ and global asymptotic stability, we combine \eqref{eq:choice_ks} and \eqref{eq:alpha_choice}, yielding
\begin{subequations}    
\begin{align} 
    \label{eq:admissible_mass}
    \frac{ k\Vert \bar{\rho}^F_x\Vert_\infty}{k_s^F-D\Vert \bar{\rho}^F_{xx}\Vert_\infty} &\le \frac{M^L - M^{FF}}{M^{FB}-M^{FF}}, \quad &\mathrm{if}\ M^{FB}>M^{FF}\\
    \label{eq:best_case_limit_ks}
    \frac{ k\Vert \bar{\rho}^F_x\Vert_\infty}{k_s^F-D\Vert \bar{\rho}^F_{xx}\Vert_\infty} &\le 1, \quad &\mathrm{if}\ M^{FB}\le M^{FF}.
\end{align}
\end{subequations}
When $M^{FB} \le M^{FF}$, the action is purely feedback and \eqref{eq:best_case_limit_ks} reduces to $k_s^F > D\Vert \bar{\rho}^F_{xx}\Vert_\infty + k\Vert \bar{\rho}^F_x\Vert_\infty$. When $M^{FB} > M^{FF}$, inequality \eqref{eq:admissible_mass} provides an explicit lower bound on the leader mass $M^L$ for the control problem to admit a solution
\begin{equation}
    M^L\ge \sup_{t\in\mathbb{R}_{\geq 0}}\left(M^{FF}+\frac{k\Vert \bar{\rho}^F_x\Vert_\infty(M^{FB}-M^{FF})}{k_s^F-D\Vert \bar{\rho}^F_{xx}\Vert_\infty}\right).
\end{equation}
Since $q^F$ is linear in $k_s^F$ and $M^{FB}$ maps to $q^F$ through linear operators, we write $M^{FB} = k_s^F M^s + M$, where $M^s$ and $M$ are the mass contributions from the switching term $-\mathrm{sign}(e)$ and the remaining terms in $q^F$, respectively. In the limit $k_s^F \to \infty$, the minimum leader mass becomes
\begin{equation} \label{eq:minimum_mass}
     M^L\ge \sup_{t\in\mathbb{R}_{\geq 0}}\left(M^{FF}+k\Vert \bar{\rho}^F_x\Vert_\infty M^{s}\right)
\end{equation}

\begin{remark}
The lower bound \eqref{eq:minimum_mass} can be 
interpreted as a \emph{herdability} condition 
\cite{lama2024shepherding} in a continuum setting, as it quantifies the minimum 
leader resources required to steer the follower 
population to a prescribed density. Additionally, condition \eqref{eq:minimum_mass} shows that this requirement grows linearly with the perturbation 
bound $k$.
\end{remark}

\subsection{Leaders' control} 
\label{sec:leaders_control}
We choose $u$ in \eqref{eq:leaders} such that leaders track the time-varying reference $\bar{\rho}^L$ in \eqref{eq:rho_hat_L} with $\alpha$ and $C$ as in Sec.~\ref{sec:deconv}. To reject the bounded perturbations on the leaders, we select $u$ such that
\begin{align}
    \left[\rho^L(x,t) u(x,t)\right]_x = q^L(x,t),
\end{align}
where
\begin{align}\label{eq:ql}
    q^L(x, t) = -k_p^Le^L(x, t) - k_s^L\mathrm{sign}\left[e^L(x, t)\right] + \bar{\rho}^L_t(x, t)+ \delta(t),
\end{align}
$e^L = \bar{\rho}^L - \rho^L$, $k_p^L(t)>0$, $k_s^L > k\Vert \bar{\rho}^L(\cdot, t)\Vert_\infty$, and $\delta$ is chosen to satisfy boundary conditions. This control law guarantees global exponential stability in $\mathcal{L}^2(\Omega)$, since we can write
\begin{align}\label{eq:leaders_stability}
    V^L_t(t) \leq -k_p^LV^L(t),
\end{align}
where $V^L = \frac{1}{2}\Vert e^L\Vert_2^2$; the proof follows from \cite[Theorem~1]{maffettone2026robust}.
\begin{figure*}[ht]
    \centering
    \includegraphics[width=1\linewidth]{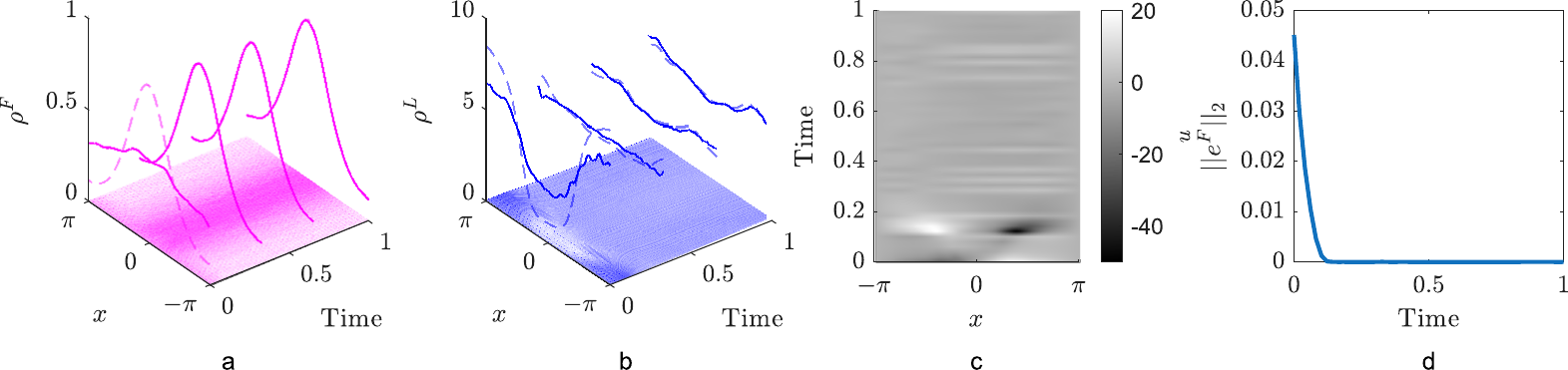}
    \caption{\textbf{Control of biased random walkers moving on a ring}.  \textbf{a}. Evolution in time and space of all the followers in the ensemble (x- and y-axes). On the z-axis the estimated (solid) and desired (dashed) densities are displayed in four representative time instants. 
    \textbf{b}. Evolution in time and space of all the leaders in the ensemble (x- and y-axes). On the z-axis the estimated (solid) and desired (dashed) densities are displayed in four representative time instants. 
    \textbf{c}. Evolution in time and space of the velocity field $u(x,t)$ generated by the controller and used to steer the leaders towards their desired density.
    \textbf{d}. Evolution of the $\mathcal{L}^2(\Omega)$ norm of the followers control error in time.
    }
    \label{fig:numerical_Validation}
\end{figure*}

\subsection{Coupled stability analysis}\label{sec:coupled_stability}
The analysis in Secs.~\ref{sec:Followers_behavior} and \ref{sec:leaders_control} treats leaders and followers independently. We now account for the coupling: the followers' Lyapunov derivative depends on the leaders' tracking error $e^L$ through the interaction kernel. Specifically, the $\mathcal{L}^2(\Omega)$ norms of the density errors are given by
\begin{subequations}
    \begin{multline} \label{eq:follower_with_leaders}
V^F_t \leq \left[D(1-\alpha)\Vert g_1\Vert_\infty-2D-\alpha k_p^F\right]V^F\\ 
- \int_\Omega e^F \left[\bar{\rho}^F (f*e^L)\right]_x\,\mathrm{d}x- \int_\Omega e^F \left[\rho^F (f*e^L)\right]_x\,\mathrm{d}x,
\end{multline}
\begin{equation} \label{eq:leaders_Lyap}
    V^L_t \leq -k_p^LV^L.
\end{equation}
\end{subequations}
The two additional terms in \eqref{eq:follower_with_leaders} capture the effect of the leaders' transient on the followers. We bound them as
\begin{subequations}\label{eq:bounds_TSS}
    \begin{align}
        \left\vert \int_\Omega e^F \left[\bar{\rho}^F (f*e^L)\right]_x\,\mathrm{d}x\right\vert &\leq J \sqrt{V^F}\sqrt{V^L},\\
        \left\vert \int_\Omega e^F \left[e^F (f*e^L)\right]_x\,\mathrm{d}x\right\vert&\leq S V^F\sqrt{V^L},
    \end{align}
\end{subequations}
where
\begin{subequations}\label{eq:constants}
        \begin{align}
        J &=  2\Vert \bar{\rho}^F_{x}\Vert_2 \Vert f\Vert_2 +2\Vert\bar{\rho}^F\Vert_2 \Vert f_x\Vert_2,  \\
        S &=  \sqrt{2}\Vert f\Vert_2.
    \end{align}
\end{subequations}
To recover such bounds we proceed similarly to \eqref{eq:bounds} bounding convolutions with Young's inequality \cite{axler2020measure}.

\begin{theorem}\label{thm:coupled_stability}
    Consider the leader--follower system \eqref{eq:model} with the followers' control law \eqref{eq:control_action}--\eqref{eq:q} and the leaders' control law \eqref{eq:ql}, under the hypotheses of Theorem~\ref{thm:followers_convergence}. If $k_p^L$ is chosen sufficiently large, then the followers' density error $e^F$ globally asymptotically converges to zero in $\mathcal{L}^2(\Omega)$.
\end{theorem}

\begin{proof}
Substituting the bounds \eqref{eq:bounds_TSS} into \eqref{eq:follower_with_leaders} yields
\begin{align} \label{eq:follower_with_leaders_simplified}
V^F_t \le & \left[D(1-\alpha)\Vert g_1\Vert_\infty-2D-\alpha k_p^F\right]V^F\nonumber\\ 
&+ J \sqrt{V^F}\sqrt{V^L} +  S V^F\sqrt{V^L}.
\end{align}
Letting $\epsilon = 1/k_p^L$, the coupled system becomes
\begin{subequations}
    \begin{align}
    V^F_t &= \left[D(1-\alpha)\Vert g_1\Vert_\infty-2D-\alpha k_p^F\right]V^F\nonumber\\ 
    \label{eq:followers_slow}
    &\quad+ J \sqrt{V^F}\sqrt{V^L} +  S V^F\sqrt{V^L},\\
    \label{eq:leaders_fast}
    \epsilon V^L_t & =  -V^L,
\end{align}
\end{subequations}
where, with an abuse of notation, $V^F$ and $V^L$ denote the variables of the bounding system. This is a singularly perturbed system in standard form \cite{kokotovic1999singular}. The boundary-layer system \eqref{eq:leaders_fast} has a unique globally exponentially stable equilibrium at $V^L = 0$. Substituting it into the reduced system \eqref{eq:followers_slow} yields
\begin{align*}
    V^F_t = \left[D(1-\alpha)\Vert g_1\Vert_\infty-2D-\alpha k_p^F\right]V^F,
\end{align*}
which is globally asymptotically stable under the hypotheses of Theorem~\ref{thm:followers_convergence}. By Tikhonov's theorem, for sufficiently small $\epsilon$ (equivalently, sufficiently large $k_p^L$), the trajectories of the full system converge to those of the reduced system, completing the proof.
\end{proof}

\begin{remark}
    The timescale separation requires $k_p^L$ to be large relative to the followers' convergence rate. Since $V^L$ decays as $e^{-k_p^L t}$ and the coupling enters through $\sqrt{V^L}$, a practical guideline is $k_p^L \gg 2(2D + \alpha k_p^F)$, ensuring that the leaders' transient is negligible on the followers' timescale.
\end{remark}

\section{Numerical Validation}\label{sec:numerical_validation}
For validation, we control an ensemble of biased random walkers with heterogeneous drifts on the unit circle
\begin{equation} \label{eq:1D_Random_Walkers}
    \mathrm{d}x_i^F = \left(b_i + \sum_{j=1}^{N^L} f(\{x_i^F(t),x_j^L(t)\}_\pi) \right)  \mathrm{d}t + \sqrt{2D}\mathrm{d}W_i 
\end{equation}
for $i=1,\,\dots,\,N^F$, with $N^F=5000$. With respect to \eqref{eq:micro}, $g_i(x_i(t))=b_i$ is a constant drift drawn from $\mathcal{U}([-2,2])$. We set $D=0.1$, $\ell=\pi$, and assign each follower a mass $1/N^F$, yielding $M^F=1$. Leaders are unperturbed single integrators,
\begin{equation} \label{eq:leaders_validation}
    \mathrm{d} x_i^L(t) = u_i(t)\,\mathrm{d}t,
\end{equation}
for $i=1,\dots,N^L$, with $N^L=5000$ and mass $30/N^L$ each, so that $M^L=30$. The target follower density is the Von~Mises distribution $\bar \rho ^F=\mathrm{e}^{(\kappa \cos(x-\mu))}$ with $\mu = 0$, $\kappa = 1$.

Since leaders are unperturbed, we set $k=1$, $r=0$, and choose $k_s^F = 5( D\Vert \bar{\rho}^F_{xx}\Vert_\infty +  k\Vert \bar{\rho}^F_x\Vert_\infty)$, $k_p^F = 2$, $k_s^L=0.1$. The choice $k_p^L=50$ enforces the timescale separation required by Theorem~\ref{thm:coupled_stability}. The derivative $\bar \rho^L_t$ is approximated via backward differences on an exponential moving average of $\bar \rho^L$. To suppress chattering, we regularize $\mathrm{sign}(x)$ as $\tanh(\eta x)$ with $\eta = 10^2$. The densities $\rho^F$ and $\rho^L$ are estimated from agents' positions using a micro-to-macro bridge: a normalized histogram on a grid of $150$ points, filtered with a Gaussian kernel of length of approximately $\pi/30$. The same filter is applied to $\bar{\rho}^L$ to reduce numerical instabilities. The macro-to-micro bridge assigns each leader $u_i(t) = u(x_i^L, t)$ by spatial sampling. Simulations are run in MATLAB using forward Euler (leaders, step $2\cdot10^{-6}$) and Euler--Maruyama (followers, step $2\cdot10^{-4}$).

Fig.~\ref{fig:numerical_Validation}a shows that the followers converge to the desired density in approximately $0.2$ time units. Fig.~\ref{fig:numerical_Validation}b displays the leaders tracking the time-varying reference $\bar \rho ^L$; the residual error is due to the regularization of the switching action and the numerical approximation of $\bar \rho^L_t$. The control field is shown in Fig.~\ref{fig:numerical_Validation}c. $\Vert e^F\Vert_2$ decreases monotonically to zero (Fig.~\ref{fig:numerical_Validation}d). The numerically computed minimum leader mass is $M^L\geq 17$, which is satisfied with $M^L=30$.

To assess robustness to increasing heterogeneity, we repeated the simulations with $b_i\sim \mathcal{U}([-B,B])$ for $B\in[2,20]$, setting $k_s^F = 5( D\Vert \bar{\rho}^F_{xx}\Vert_\infty +  B\Vert \bar{\rho}^F_x\Vert_\infty)$ and $k_p^F = 2$. As shown in Fig.~\ref{fig:Robustness_Mass}a, the steady-state error remains zero as long as \eqref{eq:minimum_mass} is satisfied. When heterogeneity makes the minimum mass required higher than the available mass, asymptotic convergence is lost and the error grows with $B$. A representative case with $B=20$ is shown in Fig.~\ref{fig:Robustness_Mass}b, where leaders fail to steer followers to the target profile.

\subsection{Finite population effects}
Our architecture rests on a mean-field assumption that requires populations to be sufficiently large. 
To quantify the minimum population sizes for the micro-to-macro bridges to operate reliably, we simulated system~\eqref{eq:model} under the control law of Sec.~\ref{sec:control_design}, using the same scheme and gains as in Fig.~\ref{fig:numerical_Validation}. The total masses are held fixed at $M^F = 1$ and $M^L = 30$ across all runs, so that varying $N^L$ or $N^F$ does not change the leader-to-follower mass ratio. $N^L$ ($N^F$) were sampled in the interval $[10,5000]$ using 30 samples equally spaced in logarithmic scale.

Fig.~\ref{fig:Finite_population_size} (orange line) shows the residual error as a function of the number of leaders $N^L$, with the number of followers fixed at $N^F = 1000$. A clear threshold emerges: below the critical value of $N^L \approx 130$, the leaders are too sparse for the macro-to-micro bridge to faithfully reproduce the control field, and the error consistently settles at a high value. Above this threshold, the error norm on the followers density decreases below $10^{-2}$, identifying this value as the minimum amount of leaders required to correctly reconstruct densities from agents' positions. 
Conversely, Fig.~\ref{fig:Finite_population_size} (blue line) reports the error as a function of $N^F$ with $N^L = 1000$ fixed. We find that, even in the presence of a large number of leaders, a minimum number of followers ($N_F \approx 400$) is required to reduce the norm of the error to values less than $10^{-2}$.
These observations are in line with the herdability conditions recently reported in~\cite{lama2024shepherding}. 

\begin{figure}
    \centering
    \includegraphics[width=1\linewidth]{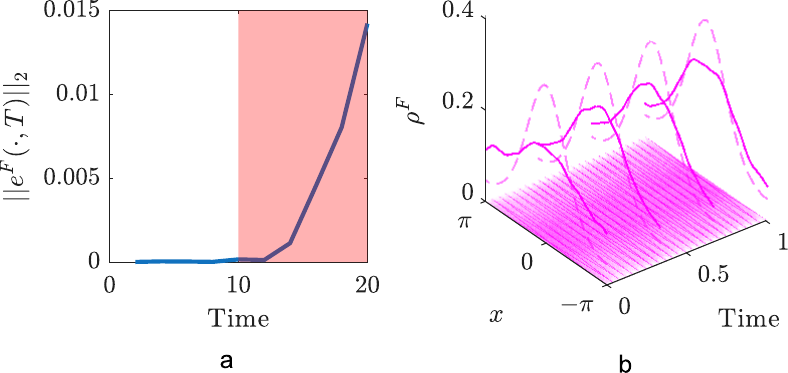}
    \caption{\textbf{Robustness to heterogeneity}.  \textbf{a}. $\mathcal{L}_2(\Omega)$ norm of the followers control error for increasing heterogeneity in the follower population. The red shaded area represents the conditions where \eqref{eq:minimum_mass} is not satisfied. $T=1$ is the terminal instant of the simulation. 
    \textbf{b}. Evolution in time and space of all the followers in the ensemble when $b_i\sim\mathcal{U}([-20,20])$. On the z axis the estimated (solid) and desired (dashed) densities are displayed in four representative time instants. 
    }
    \label{fig:Robustness_Mass}
\end{figure}

\begin{figure}
    \centering
    \includegraphics[width=1\linewidth]{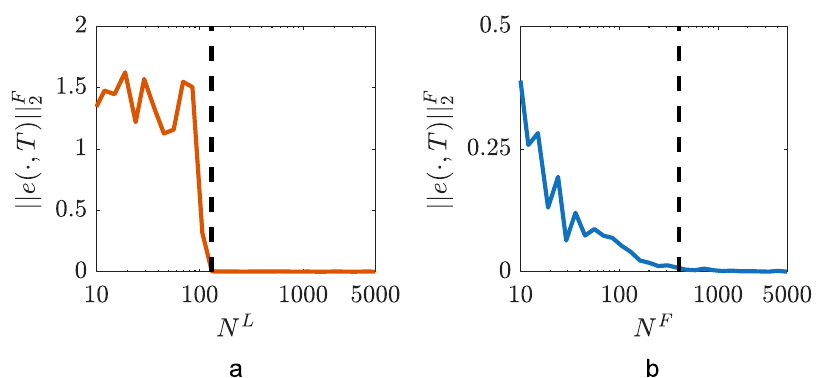}
    \caption{\textbf{Effects of finite population size}. $\mathcal{L}_2(\Omega)$ norm of the followers control error when (a) varying the number of leaders $N^L\in[10,5000]$ when $N^F=1000$, and (b) varying the number of followers $N^F\in[10,5000]$ when $N^L=1000$ (vertical dashed lines denote the threshold above which the error norm is below $10^{-2}$). $T=1.5$ is the terminal instant of the simulation. $N^L$($N^F$) were sampled in the interval with $[10,5000]$ using 30 samples equally spaced in logarithmic scale. 
    }
    \label{fig:Finite_population_size}
\end{figure}

\section{Conclusions}
\label{sec:conclusions}
We proposed a robust multi-scale leader-follower control strategy for large-scale multi-agent systems in which both populations can be affected by unknown bounded perturbations. The control architecture consists of a macroscopic feedback law that guarantees global asymptotic convergence of the followers' density to a prescribed profile, and maps it to individual leader inputs via spatial sampling. A minimum leader mass and a timescale separation between leaders and followers dynamics are required to guarantee global asymptotic stability of the followers closed loop dynamics. Numerical validation on heterogeneous biased random walkers confirms the theoretical findings and illustrates robustness to increasing levels of population heterogeneity. Furthermore, we study the effects of a finite population sizes on the performance of the control architecture.

Current limitations include the mean-field assumption ($N^F,N^L\to\infty$), the assumption of having global knowledge of both density profiles (centralized control), and the restriction to one-dimensional domains. Future work will address convergence guarantees in the presence of finite population sizes, distributed density estimation algorithms~\cite{dilorenzo2025distributed}, and higher-dimensional settings.

\section*{Acknowledgements}
AI tools were used for language editing. All technical content is entirely the authors’ own.

\bibliographystyle{IEEEtran}
\bibliography{references}

\end{document}